\documentclass[apa,notitlepage,twocolumn,nofootinbib,superscriptaddress]{revtex4-1}
\usepackage[shortlabels]{enumitem}
\usepackage{graphicx,epic,eepic,epsfig,latexsym,amssymb,verbatim,color,dsfont}
\usepackage{amsfonts}       
\usepackage{nicefrac}       
\usepackage{amsmath}
\usepackage{array,mathtools,multirow,bm,times,tcolorbox,relsize,booktabs}
\usepackage[T1]{fontenc}
\usepackage{qcircuit}
\usepackage{makecell}
\usepackage{pifont}
\usepackage{bbding}
\usepackage{bbm}
\usepackage{float}
\usepackage{tikz}
\usetikzlibrary{chains}
\usetikzlibrary{fit}
\usepackage{pgflibraryarrows}		
\usepackage{pgflibrarysnakes}		
\usepackage{epsfig}
\usetikzlibrary{shapes.symbols,patterns} 
\usepackage{pgfplots}
\usepackage[strict]{changepage}
\usepackage{hyperref}
\hypersetup{colorlinks=true,citecolor=blue,linkcolor=blue,filecolor=blue,urlcolor=blue,breaklinks=true}

\usepackage[marginal]{footmisc}
\usepackage{url}
\usepackage{theorem}

\newtheorem{proposition}{Proposition}
\newtheorem{lemma}[proposition]{Lemma}

\newtheorem{theorem}[proposition]{Theorem}

\def\squareforqed{\hbox{\rlap{$\sqcap$}$\sqcup$}}
\def\qed{\ifmmode\squareforqed\else{\unskip\nobreak\hfil
\penalty50\hskip1em\null\nobreak\hfil\squareforqed
\parfillskip=0pt\finalhyphendemerits=0\endgraf}\fi}
\def\endenv{\ifmmode\;\else{\unskip\nobreak\hfil
\penalty50\hskip1em\null\nobreak\hfil\;
\parfillskip=0pt\finalhyphendemerits=0\endgraf}\fi}
\newenvironment{proof}{\noindent \textbf{{Proof~} }}{\hfill $\blacksquare$}

\newcounter{remark}

\newcounter{example}
\newenvironment{example}[1][]{\refstepcounter{example}\par\medskip\noindent%
\textbf{Example~\theexample #1} }{\medskip}

\mathchardef\ordinarycolon\mathcode`\:
\mathcode`\:=\string"8000
\def\vcentcolon{\mathrel{\mathop\ordinarycolon}}
\begingroup \catcode`\:=\active
  \lowercase{\endgroup
  \let :\vcentcolon
  }

\usepackage{cleveref}
\usepackage{xcolor}

\newcommand{\nc}{\newcommand}
\nc{\rnc}{\renewcommand}
\nc{\lbar}[1]{\overline{#1}}
\nc{\bra}[1]{\langle#1|}
\nc{\ket}[1]{|#1\rangle}
\nc{\ketbra}[2]{|#1\rangle\!\langle#2|}
\nc{\braket}[2]{\langle#1|#2\rangle}

\nc{\proj}[1]{| #1\rangle\!\langle #1 |}
\nc{\avg}[1]{\langle#1\rangle}
\nc{\rank}{\operatorname{Rank}}
\nc{\smfrac}[2]{\mbox{$\frac{#1}{#2}$}}
\nc{\tr}{\operatorname{Tr}}
\nc{\ox}{\otimes}
\nc{\dg}{\dagger}
\nc{\cH}{{\cal H}}
\nc{\cL}{{\cal L}}
\nc{\cD}{{\cal D}}
\nc{\cS}{{\cal S}}
\nc{\bu}{{\mathbf{u}}}

\nc{\supp}{{\operatorname{supp}}}

\usepackage{hyperref}
\hypersetup{colorlinks=true,citecolor=blue,linkcolor=blue,filecolor=blue,urlcolor=blue,breaklinks=true}

\makeatletter
\def\grd@save@target#1{%
  \def\grd@target{#1}}
\def\grd@save@start#1{%
  \def\grd@start{#1}}
\tikzset{
  grid with coordinates/.style={
    to path={%
      \pgfextra{%
        \edef\grd@@target{(\tikztotarget)}%
        \tikz@scan@one@point\grd@save@target\grd@@target\relax
        \edef\grd@@start{(\tikztostart)}%
        \tikz@scan@one@point\grd@save@start\grd@@start\relax
        \draw[minor help lines,magenta] (\tikztostart) grid (\tikztotarget);
        \draw[major help lines] (\tikztostart) grid (\tikztotarget);
        \grd@start
        \pgfmathsetmacro{\grd@xa}{\the\pgf@x/1cm}
        \pgfmathsetmacro{\grd@ya}{\the\pgf@y/1cm}
        \grd@target
        \pgfmathsetmacro{\grd@xb}{\the\pgf@x/1cm}
        \pgfmathsetmacro{\grd@yb}{\the\pgf@y/1cm}
        \pgfmathsetmacro{\grd@xc}{\grd@xa + \pgfkeysvalueof{/tikz/grid with coordinates/major step}}
        \pgfmathsetmacro{\grd@yc}{\grd@ya + \pgfkeysvalueof{/tikz/grid with coordinates/major step}}
        \foreach \x in {\grd@xa,\grd@xc,...,\grd@xb}
        \node[anchor=north] at (\x,\grd@ya) {\pgfmathprintnumber{\x}};
        \foreach \y in {\grd@ya,\grd@yc,...,\grd@yb}
        \node[anchor=east] at (\grd@xa,\y) {\pgfmathprintnumber{\y}};
      }
    }
  },
  minor help lines/.style={
    help lines,
    step=\pgfkeysvalueof{/tikz/grid with coordinates/minor step}
  },
  major help lines/.style={
    help lines,
    line width=\pgfkeysvalueof{/tikz/grid with coordinates/major line width},
    step=\pgfkeysvalueof{/tikz/grid with coordinates/major step}
  },
  grid with coordinates/.cd,
  minor step/.initial=.2,
  major step/.initial=1,
  major line width/.initial=2pt,
}
\makeatother

\usepackage{thmtools}
\usepackage{thm-restate}
\usepackage{etoolbox}
\makeatletter
\def\problem@s{}
\newcounter{problems@cnt}

\newcommand{\allproblems}{\problem@s}
\makeatother
\usepackage[ruled, vlined, linesnumbered]{algorithm2e}

\allowdisplaybreaks

\begin{document}
\title{Limitations of Classically-Simulable Measurements for Quantum State Discrimination}
\author{Chengkai Zhu}
\thanks{Chengkai Zhu and Zhiping Liu contributed equally to this work.}
\affiliation{Thrust of Artificial Intelligence, Information Hub,\\
The Hong Kong University of Science and Technology (Guangzhou), Guangzhou 511453, China}
\author{Zhiping Liu}
\thanks{Chengkai Zhu and Zhiping Liu contributed equally to this work.}

\affiliation{Thrust of Artificial Intelligence, Information Hub,\\
The Hong Kong University of Science and Technology (Guangzhou), Guangzhou 511453, China}
\affiliation{National Laboratory of Solid State Microstructures, School of Physics and Collaborative Innovation Center of Advanced Microstructures, Nanjing University, Nanjing 210093, China}
\author{Chenghong Zhu}
\author{Xin Wang}
\email{felixxinwang@hkust-gz.edu.cn}
\affiliation{Thrust of Artificial Intelligence, Information Hub,\\
The Hong Kong University of Science and Technology (Guangzhou), Guangzhou 511453, China}

\begin{abstract}
In the realm of fault-tolerant quantum computing, stabilizer operations play a pivotal role, characterized by their remarkable efficiency in classical simulation. This efficiency sets them apart from non-stabilizer operations within the quantum computational theory. In this Letter, we investigate the limitations of classically-simulable measurements in distinguishing quantum states. We demonstrate that any pure magic state and its orthogonal complement of odd prime dimensions cannot be unambiguously distinguished by stabilizer operations, regardless of how many copies of the states are supplied. We also reveal intrinsic similarities and distinctions between the quantum resource theories of magic states and entanglement in quantum state discrimination. The results emphasize the inherent limitations of classically-simulable measurements and contribute to a deeper understanding of the quantum-classical boundary.
\end{abstract}

\date{\today}
\maketitle

\emph{Introduction.---}
The computational power of quantum computers, including a substantial speed-up over their classical counterparts in solving certain number-theoretic problems~\cite{Shor_1997,grover1996fast,RevModPhysAndrew} and simulating quantum systems~\cite{lloyd1996universal,Childs_2018}, can only be unlocked with a scalable quantum computing solution. Fault-tolerant quantum computation (FTQC) provides a scheme to overcome obstacles of physical implementation such as decoherence and inaccuracies~\cite{shor1997faulttolerant,Campbell_2017,knill2005quantum}.

A cornerstone of the FTQC resides in stabilizer circuits, comprised exclusively of the Clifford gates. It is well-known that the stabilizer circuits can be efficiently classically simulated~\cite{gottesman1997stabilizer}, and therefore do not confer any quantum computational advantage. However, \textit{magic states} are quantum states that cannot be prepared using the stabilizer formalism~\cite{Veitch2014}, and can promote the stabilizer circuits to universal quantum computation via state injection~\cite{Zhou_2000,Gottesman_1999,Bravyi_2005}. In this context, the magic states and non-stabilizer operations characterize the computational power of universal quantum computation. 

While extensive research has explored the stabilizerness of quantum states and gates within circuits~\cite{Campbell_2012, Yoganathan_2019}, a crucial yet underexplored facet is the stabilizerness of quantum measurements~\cite{nielsen2010quantum} -- a critical process for reliably decoding classical information encoded in quantum states. In general, it is not applicable for one to access the physical properties of a locally interacting quantum many-body system by classical simulation. However, when information is encoded in a stabilizer state, the decoding process via stabilizer measurements remains efficiently classically simulable~\cite{Mari2012}. This prompts a fundamental question: can stabilizer measurements perfectly decode all tasks, or are there inherent limitations? Investigating the distinction in decoding capabilities between stabilizer measurements, which are classically efficiently simulable, and other measurements becomes paramount for understanding the intricate relationship between classical information encoded in quantum states and the measurement process.

The ability to retrieve classical information from quantum systems varies significantly with different measurements. One celebrated example is the quantum nonlocality without entanglement~\cite{Bennett1999b}. In essence, global measurements can always perfectly distinguish mutually orthogonal quantum states, while there is a set of product states that cannot be distinguished via local quantum operations and classical communications (LOCC). This distinction between global and local measurements has garnered substantial attention, proving to be intricately linked with quantum entanglement theory and the concept of nonlocality~\cite{Leung_2021,Bandyopadhyay2015LimSep,Childs2013,Bandyopadhyay2011,Calsamiglia2010,Halder2019,Walgate2000,Bennett1999}. This primitive gap between distinct classes of measurements makes quantum state discrimination (QSD) a crucial aspect of fundamental physics~\cite{bae2015quantum,Correa_2022}, where it can be used to test the principles and nature of quantum mechanics. Moreover, QSD has led to fruitful applications in quantum cryptography~\cite{Gisin_2002, Cleve_1999, Leverrier_2009}, quantum dimension witness~\cite{Brunner_2013, Hendrych_2012} and quantum data hiding~\cite{Terhal2001,Eggeling2002,Matthews2009}. 

Inspired by the intrinsic behavior of different measurements in entanglement theory, we raise a natural and important question for understanding the limit and power of the classically-simulable measurements. In particular, is there a sharp gap between the classically-simulable measurements and those that could potentially promote universal quantum computation? If such a gap exists, it will imply considerable advantages that the resource of magic states can provide to the measurement in quantum information processing.

In this Letter, we give an affirmative answer to this question. We show that any pure magic state and its orthogonal complement cannot be unambiguously distinguished via Positive Operator-Valued Measures (POVMs) having positive discrete Wigner functions, which are classically-simulable and strictly include stabilizer measurements~\cite{Mari2012, Veitch2012}, no matter how many copies of the states are supplied. We also demonstrate an exponential decay on the asymptotic minimal error probability for distinguishing the Strange state and its orthogonal complement via POVMs having positive discrete Wigner functions, where the Strange state is a representative qutrit magic state defined as $\ket{\mathbb{S}} \coloneqq (\ket{1}- \ket{2})/\sqrt{2}$~\cite{Veitch2014}. 

In addition, we show that every set of orthogonal pure stabilizer states can be unambiguously distinguished via POVMs having positive discrete Wigner functions, indicating there is no similar phenomenon as the unextendible product basis (UPB) in entanglement theory. Moreover, we demonstrate that even with the assistance of one or two copies of any qutrit magic state, the Strange state and its orthogonal complement remain indistinguishable via POVMs having positive discrete Wigner functions. It is different from entanglement theory where a single copy of the Bell state is always sufficient to perfectly distinguish a pure entangled state and its orthogonal complement using PPT POVMs~\cite{Yu2014}.

\emph{Preliminaries.---}
To characterize the stabilizerness of quantum states and operations, we first recall the definition of the discrete Wigner function~\cite{WOOTTERS19871,Gross_2006a,Gross_2006b}. Throughout the paper, we study the Hilbert space $\cH_d$ with an odd dimension $d$, and if the dimension is not prime, it should be understood as a tensor product of Hilbert spaces each having an odd prime dimension. Let $\cL(\cH_d)$ be the space of linear operators mapping $\cH_d$ to itself and $\cD(\cH_d)$ be the set of density operators acting on $\cH_d$. It is worth noting that qudit-based quantum computing is gaining increasing significance, as numerous problems in the field are awaiting further exploration~\cite{Wang_2020}.

Given a standard computational basis $\{\ket{j} \}_{j=0,\cdots,d-1}$, the unitary boost and shift operators $X, Z \in \cL(\cH_d)$ are defined by $X\ket{j} = \ket{j \oplus 1}, Z\ket{j} = w^j \ket{j}$, where $w = e^{2\pi i/d}$ and $\oplus$ is the addition in $\mathbb{Z}_d$. The \textit{discrete phase space} of a single $d$-level system is $\mathbb{Z}_d \times \mathbb{Z}_d$. At each point $\mathbf{u}=(a_1,a_2) \in \mathbb{Z}_d \times \mathbb{Z}_d$,
the discrete Wigner function of a state $\rho$ is defined as $W_\rho(\mathbf{u}):=\frac{1}{d} \operatorname{Tr}\left[A_{\mathbf{u}} \rho\right]$ where $A_{\mathbf{u}}$ is the phase-space point operator given by $A_\mathbf{0}:=\frac{1}{d} \sum_{\mathbf{u}} T_{\mathbf{u}}, A_{\mathbf{u}}:=T_{\mathbf{u}} A_\mathbf{0} T_{\mathbf{u}}^{\dagger}$ and $T_{\mathbf{u}}=\tau^{-a_1 a_2} Z^{a_1} X^{a_2}, \tau=e^{(d+1) \pi i / d}$. We say a state $\rho$ has positive discrete Wigner functions (PWFs) if $W_{\rho}(\mathbf{u}) \geq 0, \forall \mathbf{u} \in \mathbb{Z}_d \times \mathbb{Z}_d$ and briefly call it PWF state. Let $\mathbf{E} = \{ E_j\}_{j=0}^{n-1}$ be an $n$-valued POVM acting on $\cH_d$ with $\sum_{j=0}^{n-1}E_j = \mathds{1}$. 
The discrete Wigner function of each effect $E_j$ is given by $W(E_j|\mathbf{u}) = \tr[E_jA_\mathbf{u}]$. $\mathbf{E}$ is said to be a PWF POVM if each $E_j$ has PWFs. More details can be found in appendix.

In odd prime dimensions, quantum circuits with initial states and all subsequent quantum operations having PWFs, which strictly include stabilizer (STAB) operations, admit efficient classical simulations~\cite{Mari2012, Veitch2012}, extending the Gottesman-Knill theorem. On the contrary, negativity in Wigner functions is usually regarded as an indication of `nonclassicality'~\cite{Galv_o_2005, Cormick_2006} and identified as a computational resource. Thus, PWF POVMs are recognized as classically-simulable measurements~\cite{Howard_2014}. The exclusive applicability of these results to odd prime dimensions may stem from the unique property that only quantum systems of such dimensions exhibit covariance of the Wigner function w.r.t. Clifford operations~\cite{Huangjun_2016}. It's worth noting that there exist mixed magic states with PWFs, rendering them useless for magic state distillation~\cite{Bravyi_2005}. These states are termed \textit{bound universal states}~\cite{Campbell_2010}, analogous to states with a positive partial transpose (PPT) in entanglement distillation~\cite{Horodecki_1998}. Therefore, PWF POVMs strictly include all STAB POVMs as
\begin{equation*}
     \text{STAB POVMs} \subsetneq \textrm{PWF POVMs} \subsetneq \textrm{All POVMs}.
\end{equation*}

\emph{Asymptotic limits of PWF POVMs for a pure state and its orthogonal complement.---}
Our primary aim is to elucidate the constraints inherent in measurements that can be efficiently classically simulated. QSD describes a general process of extracting classical information from quantum systems via measurements. To distinguish two states, one usually performs a two-outcome POVM on the received state and then determines which state it is according to the measurement outcome. 

It is well-known that the asymptotic regime of QSD can unravel the underlying mechanism of entanglement~\cite{Walgate2000,Bandyopadhyay2011,cheng2023discrimination}. The limit of local measurements exhibits a fundamental distinction between pure and mixed states~\cite{Bandyopadhyay2011}. Moreover, the asymptotic error probability in QSD is interlinked with the quantum relative entropy, Petz's R\'enyi divergence~\cite{Audenaert_2008}, and the sandwiched R\'enyi divergence~\cite{Ogawa_2005, Mosonyi_2014}. Notably, in the regime of many copies, greater flexibility and options exist for the potential POVMs. However, we shall show a wide range of quantum states that cannot be unambiguously distinguished via PWF POVMs, including STAB POVMs, no matter how many copies are supplied.

\begin{theorem}\label{thm:main_result}
Let $\rho_0 \in \cD(\cH_d)$ be a pure magic state and $\rho_1=(\mathds{1}-\rho_0)/(d - 1)$ be its orthogonal complement, where $\mathds{1}$ is the identity matrix. Then for any integer $n\in \mathbb{Z}^+$, $\rho_0^{\ox n}$ and $\rho_1^{\ox n}$ cannot be unambiguously distinguished by PWF POVMs.
\end{theorem}

Theorem~\ref{thm:main_result} reveals a significant disparity in the ability of PWF POVMs and other measurements in QSD. It indicates that the classical information you are allowed to extract from the encoded states is limited when the measurements allowed are restricted to those classically-simulable ones. 
The limitation of the classically-simulable measurements cannot be overcome even by increasing the number of copies of the states.

From the angle of quantum resource theories (QRTs)~\cite{RMP_Chitambar_2019}, this theorem unravels the challenge of distinguishing a pure resourceful state and its orthogonal complement via free operations in the QRT of magic states. This parallels a phenomenon in entanglement theory where any pure entangled state and its orthogonal complement cannot be unambiguously distinguished via PPT POVMs with an arbitrary number of copies provided~\cite{Yu2014,li2017indistinguishability, cheng2023discrimination}. However, perfect distinguishability is achievable through global measurements. Notably, Takagi and Regula introduced a quantifier of resourcefulness for measurements~\cite{Takagi_2019}, demonstrating that resourceful measurements can outperform free measurements in certain QSD tasks~\cite{Oszmaniec_2019}. Here, our result further specifies the constraints of free measurements within the QRT of magic states, revealing that free operations cannot distinguish a pure resourceful state and its orthogonal complement, even in the many-copy regime.

The proof of Theorem~\ref{thm:main_result} relies on Lemma~\ref{lem:unext_to_strongly_unext} which identifies the feature of \textit{PWF unextendible} subspaces, and a fact that the orthogonal complement of any pure magic state is PWF since $-1/d\leq W_{\rho}(\mathbf{u})\leq 1/d, \forall \rho \in \cD(\cH_d), \forall \mathbf{u} $~\cite{Huangjun_2016}. We call a subspace $\cS \subseteq \cH_d$ \textit{PWF unextendible} if there is no PWF state $\rho$ whose support is a subspace of $\cS^\perp$, and \textit{PWF extendible} otherwise. A subspace $\cS \subseteq \cH_d$ is called \textit{strongly PWF unextendible} if for any positive integer $n$, $\cS^{\ox n}$ is PWF unextendible. As a simple example, if we let $\cS^\perp$ be a one-dimensional subspace spanned by the strange state $\ket{\mathbb{S}}$, then $\cS$ is (strongly) PWF unextendible.
In fact, the unextendibility of subspaces indicates the distinguishability of quantum states. It is well-known that a UPB for a multipartite quantum system indicates indistinguishability under LOCC operations~\cite{Bennett1999}.

\begin{lemma}\label{lem:unext_to_strongly_unext}
    For a PWF unextendible subspace $\cS \subseteq \cH_d$, if there is a PWF state $\rho \in \cD(\cS)$ such that $\mathrm{supp}(\rho) = \cS$, then $\cS$ is strongly PWF unextendible.
\end{lemma}
We note that Lemma~\ref{lem:unext_to_strongly_unext} implies that for a set of orthogonal quantum states $\{\rho_1,...,\rho_n\}$, if there is a $\rho_i$ whose support is strongly PWF unextendible, then $\{\rho_1,...,\rho_n\}$ cannot be unambiguously distinguished by PWF POVMs no matter how many copies are used. This leads to and generalizes the result of Theorem~\ref{thm:main_result}. We sketch the proof of Lemma~\ref{lem:unext_to_strongly_unext} as follows. 

First, we demonstrate that $\cS^{\ox 2}$ is PWF unextendible through a proof by contradiction. Suppose $\rho_s\in \cD(\cS)$ is a PWF state such that $\mathrm{supp}(\rho_s) = \cS$. If there is a PWF state $\sigma$ supporting on $(\cS^{\ox 2})^\perp$, then we have $\tr[\sigma(\rho_s\ox \rho_s)] = 0$ which leads to
$\tr[\rho_s\tr_2[\sigma(\mathds{1}\ox \rho_s )] ] = 0$. It is easy to check that $\sigma'= \tr_2[\sigma(\mathds{1}\ox \rho_s )]$ is a positive semi-definite operator with PWFs if it is non-zero. If it is zero, we can check that $\tr_1\sigma$ is a positive semi-definite operator with PWFs. In either case, we will get a PWF state supported on $\cS^\perp$, a contradiction to the PWF unextendibility of $\cS$. Hence, we conclude that $\cS^{\ox 2}$ is PWF unextendible. Using a similar technique, we can conclusively demonstrate that $\cS^{\ox n}$ is PWF unextendible for any positive integer $n$. The details can be found in appendix.

\vspace{2mm}
\emph{Asymptotic limits of PWF POVMs for mixed states.---}
Followed by Lemma~\ref{lem:unext_to_strongly_unext}, we note that Theorem~\ref{thm:main_result} displays a special case of a strongly PWF unextendible subspace. The orthogonal complement of a pure magic state turns out to be a PWF state which lies in a $d-1$ dimensional PWF unextendible subspace. This prompts an intriguing inquiry into the minimal dimension of such subspace. Notably, we will show there is a much smaller strongly PWF unextendible subspace, indicating the presence of mixed magic states that cannot be unambiguously distinguished from their orthogonal complements via PWF POVMs in the many-copy scenario. 

\begin{proposition}\label{prop:case_allmagic}
There exists a strongly PWF unextendible subspace $\cS\subseteq \cH_d$ of dimension $(d+1)/2$.
\end{proposition}
This proposition implies there is a $(d-1)/2$-dimensional subspace in which all states are magic states. The detailed proof is in appendix and we give a simple example as follows.
\begin{example}
Consider a qudit system with $d=5$. We have the following basis that spans $\cH_5$.
\begin{equation}
\begin{aligned}
&\ket{v_0} = \ket{0},\\
&\ket{v_1} = (\ket{1} + \ket{2} + \ket{3} + \ket{4})/2,\\
&\ket{v_2} = (-\ket{1} + \ket{2} + \ket{3} - \ket{4})/2,\\
&\ket{v_3} = (\ket{1} - \ket{2} + \ket{3} - \ket{4})/2,\\
&\ket{v_4} = (\ket{1} + \ket{2} - \ket{3} - \ket{4})/2.
\end{aligned}
\end{equation}
Let $\rho_0 = (\ketbra{v_0}{v_0} + \ketbra{v_1}{v_1} + \ketbra{v_2}{v_2})/3, \rho_1 = (\ketbra{v_3}{v_3} + \ketbra{v_4}{v_4})/2$, and $\cS_0 = \mathrm{supp}(\rho_0), \cS_1 = \mathrm{supp}(\rho_1)$. Followed by the idea in the proof of Proposition~\ref{prop:case_allmagic}, one can check that there is no PWF state in $\cS_1$, and $\rho_0$ is a PWF state. Thus, $\cS_0$ is a strongly PWF unextendible subspace. $\rho_0$ and $\rho_1$ cannot be unambiguously distinguished by PWF POVMs, no matter how many copies of them are supplied. 
\end{example}

More generally, we establish an easy-to-compute criterion for identifying the circumstances under which two quantum states cannot be unambiguously distinguished by PWF POVMs in the many-copy scenario. 
\begin{theorem}\label{thm:criterion_nm}
    Given $\rho_0,\rho_1\in \cD(\cH_d)$, if any of them has strictly positive discrete Wigner functions, i.e., $W_{\rho_i}(\mathbf{u}) >0, \forall \mathbf{u}$, then for any integer $n\in \mathbb{Z}^+$, $\rho_0^{\ox n}$ and $\rho_1^{\ox n}$ cannot be unambiguously distinguished by PWF POVMs.
\end{theorem}
Theorem~\ref{thm:criterion_nm} is of broad applicability for both pure and mixed states. The indistinguishability can be checked through a simple computation of the discrete Wigner functions, streamlining the conventional method by analyzing exponentially large Hilbert space.

\emph{Minimum error discrimination by PWF POVMs.---}
After characterizing the limits of PWF POVMs, we further study the minimum error QSD to unveil the capabilities inherent in PWF POVMs. For states $\rho_0$ and $\rho_1$ with prior probability $p$ and $1-p$, respectively, we denote $P_{\mathrm{e}}^{{\textup{\tiny PWF}}}(\rho_0, \rho_1, p)$ as the optimal error probability of distinguishing them by PWF POVMs. Mathematically, this optimal error probability can be expressed via semidefinite programming (SDP)~\cite{boyd2004convex} as follows. 
\begin{subequations}\label{SDP:min_error}
\begin{align}
P_{\mathrm{e}}^{{\textup{\tiny PWF}}} = \min_{E_0,E_1} &\; (1-p)\tr(E_0\rho_1) + p\tr(E_1\rho_0),\\
 {\rm s.t.} & \;\; E_0 \geq 0, E_1 \geq 0, E_0+E_1=\mathds{1}, \\
            & \;\; W(E_0 | \mathbf{u}) \geq 0, W(E_1 | \mathbf{u}) \geq 0, \forall \mathbf{u},\label{Eq:NMPOVM_const}
\end{align}
\end{subequations}
where Eq.~\eqref{Eq:NMPOVM_const} ensures $\{E_0, E_1\}$ is a PWF POVM. We provide the dual SDP in appendix. For $\rho_0$ to be the Strange state and $\rho_1$ to be its orthogonal complement, we demonstrate the following asymptotic error behavior.

\begin{proposition}\label{prop:s_state_err}
Let $\rho_0$ be the Strange state $\ketbra{\mathbb{S}}{\mathbb{S}}$ and $\rho_1 = (\mathds{1}-\ketbra{\mathbb{S}}{\mathbb{S}})/2$ be its orthogonal complement. For $n\in \mathbb{Z}^+$, we have 
\begin{equation}
    P_{\mathrm{e}}^{{\textup{\tiny PWF}}}(\rho_0^{\ox n}, \rho_1^{\ox n}, \frac{1}{2}) = \frac{1}{2^{n+1}}.
\end{equation}
The optimal PWF POVM is $\{E, \mathds{1}-E \}$, where
$E = (\ketbra{\mathbb{K}}{\mathbb{K}} + \ketbra{\mathbb{S}}{\mathbb{S}})^{\ox n}$ and $\ket{\mathbb{K}} = (\ket{1}+\ket{2})/\sqrt{2}$.
\end{proposition}

We remark what we obtain here is the optimal error probability using PWF POVMs to distinguish $n$ copies of the Strange state and its orthogonal complement. We first find the protocol above for the desired error probability and then utilize the dual SDP of~\eqref{SDP:min_error} to establish the optimality of this protocol. The detailed proof is provided in appendix. It can be seen that the optimal error probability will exponentially decay with respect to the number of copies supplied. Nevertheless, it is important to note that the error persists for all finite values of $n$, aligning with the indistinguishability established in Theorem~\ref{thm:main_result}.

We further discuss the relationship between Proposition~\ref{prop:s_state_err} and the Chernoff exponent in hypothesis testing. The celebrated quantum Chernoff theorem~\cite{Audenaert_2007,Audenaert_2008,hiai1991proper} establishes that $\xi_{C}(\rho_0,\rho_1) := \lim_{n\rightarrow \infty} - \frac{1}{n}\log P_e(\rho_0^{\ox n}, \rho_1^{\ox n}, p) = -\min_{0\leq s\leq 1}\log \tr[\rho_0^{1-s}\rho_1^{s}]$, where $P_e(\rho_0^{\ox n}, \rho_1^{\ox n}, p)$ is the average error of distinguishing $\rho_0$ and $\rho_1$ via global measurements, $\xi_{C}(\rho_0, \rho_1)$ is the so-called Chernoff exponent. The Chernoff exponent concerning a specific class of measurements, e.g., $\{\mathrm{LOCC}, \mathrm{PPT}, \mathrm{SEP}\}$, is defined in~\cite{cheng2023discrimination}. The authors proved that the Chernoff bounds in these cases are indeed faithful by showing an exponential decay of $P_e^{X}(\rho_0,\rho_1,p)$ where $X\in \{\mathrm{LOCC}, \mathrm{PPT}, \mathrm{SEP}\}$. Similarly, Proposition~\ref{prop:s_state_err} may give an insight that the Chernoff bound concerning PWF measurements is also faithful.

Proposition~\ref{prop:s_state_err} also implies applications in quantum data hiding~\cite{Terhal2001,DiVincenzo_2002,Lami_2018}. Despite the original data-hiding setting where pairs of states of a bipartite system are perfectly distinguishable via general entangled measurements yet almost indistinguishable under LOCC, it is conceivable to extend data-hiding techniques to broader contexts dictated by specific physical circumstances~\cite{Takagi_2019}. As discussed in~\cite{Takagi_2019}, one may consider the scenario that information is encoded in a way that Pauli measurements have less capability of decoding it than arbitrary measurements. Then only the party with the ability to generate magic can reliably retrieve the message. Here, we define $\|\cdot\|_{\tiny \text{PWF}}$ and $R(\text{PWF})$ as the \textit{distinguishability norm} and the \textit{data-hiding ratio}~\cite{Lami_2018} associated with PWF POVMs, respectively. Proposition~\ref{prop:s_state_err} directly gives a lower bound on the \textit{data-hiding ratio} against PWF POVMs as follows.
\begin{equation}
\label{eq:data_hiding_PWF}
    R(\text{PWF}) = \max \frac{\|p\rho-(1-p)\sigma\|_{\tiny \text{All}}}{\|p\rho-(1-p)\sigma\|_{\tiny \text{PWF}}} \geq \frac{1}{1-2^{-n}}.
\end{equation}
We also observe that a potential correlation between $R(\text{PWF})$ and the generalized robustness of measurement~\cite{Takagi_2019} merits further investigation, with preliminary evidence provided in appendix.

\emph{Distinctions between QRT of magic states and entanglement in QSD tasks.---}
The asymptotic limits of PWF POVMs share similarities with LOCC operations, both of which are considered free within their respective resource theories. Whereas, there are fundamental distinctions between the QRT of magic states and entanglement, considering the QSD tasks. In Table~\ref{tab:comparison_magic_entangle}, we display a comparison between the QRT of magic states and entanglement in QSD, including their similarities and the following distinctions. 

Recall that in entanglement theory, the UPB is an incomplete orthogonal product basis whose complementary subspace contains no product state~\cite{Bennett1999}. It shows examples of orthogonal product states that cannot be perfectly distinguished by LOCC operations. Correspondingly, we may imagine whether there is a similar `UPB' phenomenon in the QRT of magic states. That is if there is an incomplete orthogonal stabilizer basis whose complementary subspace contains no stabilizer state. We show that this is not the case as follows.

\begin{table}
    \centering
    \renewcommand\arraystretch{1.35}
    \footnotesize
    \begin{tabular}{ccc}
    \hline\hline
    & \makecell[c]{QRT of\\ magic states} & \makecell[c]{QRT of\\ entanglement} \\
    \hline
    Asymptotic limits of free POVMs & \ding{52} & \ding{52} \\
    \hline 
    Existence of UPB phenomenon & \ding{55} & \ding{52}\\ 
    \hline \\[-10pt]
    \makecell[c]{Perfect discrimination with the\\ aid of one copy of maximal resource} & \ding{55} & \ding{52}\\
    \hline\hline
    \end{tabular}
    \caption{\textbf{Comparison between the QRT of magic states and entanglement.} The second row represents if any resourceful pure state and its orthogonal complement are indistinguishable by free measurements in the many-copy scenario. The third row represents whether there is a UPB phenomenon. The last row represents whether the assistance of one copy of the maximally resourceful state is sufficient for perfect discrimination.}
    \label{tab:comparison_magic_entangle}
\end{table}

\begin{theorem}\label{thm:absence_UPB}
For a subspace $\cS\in \cH_d$, if $\cS$ has a set of basis $\{\ket{\psi_i}\}_{i=0}^{n-1}$ where every $\ket{\psi_i}$ is a stabilizer state, then $\cS$ is PWF extendible.
\end{theorem}
A direct consequence of this theorem is that any set of orthogonal pure stabilizer states $\{\ket{\psi_i}\}_{i=0}^{n-1}$ can be unambiguously distinguished via PWF POVMs as we can choose $E_i = \ketbra{\psi_i}{\psi_i}$ for $i=0,1,\cdots,n-1$ and $E_{n} = \mathds{1}-\sum_{i=0}^{n-1}\ketbra{\psi_i}{\psi_i}$. Therefore, we confirm the absence of an analogous UPB phenomenon in the QRT of magic states.

Besides, it was shown that one copy of the Bell state is always sufficient for perfectly distinguishing any pure state $\rho_0$ and its orthogonal complement $\rho_1$ via PPT POVMs~\cite{Yu2014}, i.e., distinguishing $\rho_0\ox\Phi^+_2$ and $\rho_1\ox\Phi^+_2$. However, things are different in the QRT of magic states where we find the Strange state and its orthogonal complement cannot be perfectly distinguished by PWF POVMs with the assistance of one or two copies of any qutrit magic state.

\begin{proposition}\label{prop:magic_cost_sstate}
    Let $\rho_0$ be the Strange state $\ketbra{\mathbb{S}}{\mathbb{S}}$ and $\rho_1 = (\mathds{1}-\ketbra{\mathbb{S}}{\mathbb{S}})/2$ be its orthogonal complement. $\rho_0 \otimes \tau^{\otimes k}$ and $\rho_1 \otimes \tau^{\otimes k}$ cannot be perfectly distinguished by PWF POVMs for any qutrit magic state $\tau$ and $k=1$ or $2$.
\end{proposition}
The main idea is to analyze the minimal \textit{mana}~\cite{Wang2018} $\tau^{\otimes k}$ must have to perfectly distinguish $\rho_0 \otimes \tau^{\otimes k}$ and $\rho_1 \otimes \tau^{\otimes k}$ by PWF POVMs. 
A similar result can be obtained for the Norell state $\ket{\mathbb{N}}:=(-\ket{0}+2\ket{1}-\ket{2})/\sqrt{6}$~\cite{Veitch2014}. Hence, we have witnessed the distinctions of the QRT of magic states and entanglement in regard to the resource cost for perfect discrimination.

\emph{Concluding remarks.---}
We have explored the limitations of PWF POVMs which can be efficiently classically simulated and strictly include all stabilizer measurements. Our results show that the QRT of magic states and entanglement exhibit significant similarities and distinctions in quantum state discrimination.

These results have implications in various fields, including connections between the QRT of magic states and quantum data hiding~\cite{Takagi_2019,Matthews2009,Lami_2018,DiVincenzo_2002}. It remains interesting to further study the limits of stabilizer measurements or classically-simulable ones in quantum channel discrimination~\cite{Pirandola_2019,Wang_2019,Takagi_2019a} and other operational tasks~\cite{Uola_2020, Ducuara_2020, Lipka_Bartosik_2021}. Note that as it is still open whether all operations with negative discrete Wigner functions are useful for magic state distillation~\cite{Veitch2014}, a comprehensive characterization of the quantum-classical boundary of measurements is still needed. Additionally, it is interesting to study the limitations of stabilizer measurements in a multi-qubit system~\cite{Howard2017,Bravyi2016,Seddon_2019,Hahn_2022}, and recent advances in generalized phase-space simulation methods for qubits~\cite{Zurel_2020, Raussendorf_2020} offer potential avenues to explore this, which we will leave to future work.

\emph{Acknowledgments.--} We would like to thank the anonymous referees for their helpful comments. This work was supported by the National Key R\&D Program of China (Grant No. 2024YFE0102500), the Start-up Fund (Grant No. G0101000151) from The Hong Kong University of Science and Technology (Guangzhou), the Guangdong Provincial Quantum Science Strategic Initiative (Grant No. GDZX2303007), the Education Bureau of Guangzhou Municipality, and the Guangdong Provincial Key Lab of Integrated Communication, Sensing and Computation for Ubiquitous Internet of Things (Grant No. 2023B1212010007).
\bibliography{main}

\appendix
\setcounter{subsection}{0}
\setcounter{table}{0}
\setcounter{figure}{0}

\vspace{2cm}
\onecolumngrid
\vspace{2cm}

\begin{center}
\large{Supplemental Material}
\end{center}

\renewcommand{\theequation}{S\arabic{equation}}
\renewcommand{\thesubsection}{\normalsize{Supplemental Note \arabic{subsection}}}
\renewcommand{\theproposition}{S\arabic{proposition}}
\renewcommand{\thedefinition}{S\arabic{definition}}
\renewcommand{\thefigure}{S\arabic{figure}}
\setcounter{equation}{0}
\setcounter{table}{0}
\setcounter{section}{0}
\setcounter{proposition}{0}
\setcounter{definition}{0}
\setcounter{figure}{0}

In this Supplemental Material, we provide detailed proofs of the theorems and propositions in the manuscript ``Limitations of Classically-Simulable Measurements for Quantum State Discrimination''. In Appendix~\ref{appendix:Wig_fun}, we cover the basics of the discrete Wigner function. In Appendix~\ref{appendix:asym_limits_nmPOVM}, we first present the detailed proofs for Lemma~\ref{lem:unext_to_strongly_unext} and Proposition~\ref{prop:case_allmagic}, which characterize the asymptotic limits of PWF POVMs for distinguishing a pure magic state and its orthogonal complement and a mixed magic state and its orthogonal complement, respectively. Then, we provide the proof of Theorem~\ref{thm:criterion_nm} which serves as an easy-to-compute criterion for when PWF POVMs cannot unambiguously distinguish two quantum states in the many-copy scenario. Appendix~\ref{appendix:dual_sdp_minerr} introduces the primal and dual SDP for calculating the optimal error probability of distinguishing two quantum states via PWF POVMs. We further provide detailed proof of Proposition~\ref{prop:s_state_err}. Then in Appendix~\ref{appendix:distinctions}, we furnish detailed proofs for Theorem~\ref{thm:absence_UPB} and Proposition~\ref{prop:magic_cost_sstate}, both of which characterize the distinctions between the QRT of magic states and entanglement in QSD tasks.

\section{The discrete Wigner function}\label{appendix:Wig_fun}
We denote $\cH_d$ as a Hilbert space of dimension $d$, and $\{ \ket{j} \}_{j = 0, \cdots, d-1}$ as the standard computational basis. Let $\cL(\cH_d)$ be the space of operators mapping $\cH_d$ to itself. For odd prime dimension $d$, the unitary boost and shift operators $X, Z \in \cL(\cH_d)$ are defined as~\cite{WWS19}: 
\begin{equation}
X\ket{j} = \ket{j \oplus 1},\quad Z\ket{j} = w^j \ket{j},
\end{equation}
where $w = e^{2\pi i/d}$ and $\oplus$ denotes addition modulo $d$. The \textit{discrete phase space} of a single $d$-level system is $\mathbb{Z}_d \times \mathbb{Z}_d$, which can be associated with a $d\times d$ cubic lattice. For a given point in the discrete phase space $\mathbf{u}=\left(a_1, a_2\right) \in \mathbb{Z}_d \times \mathbb{Z}_d$, the Heisenberg-Weyl operators are given by \begin{equation}
    T_{\mathbf{u}}=\tau^{-a_1 a_2} Z^{a_1} X^{a_2},
\end{equation}
where $\tau=e^{(d+1) \pi i / d}$. These operators form a group, the Heisenberg-Weyl group, and are the main ingredient for representing quantum systems in finite phase space. The case of non-prime dimension can be understood to be a tensor product of $T_\mathbf{u}$ with odd prime dimension. For each point $\mathbf{u} \in \mathbb{Z}_d \times \mathbb{Z}_d$ in the discrete phase space, there is a phase-space point operator $A_{\mathbf{u}}$ defined as 
\begin{equation}
A_\mathbf{0}:=\frac{1}{d} \sum_{\mathbf{w}} T_{\mathbf{w}},\quad A_{\mathbf{u}}:=T_{\mathbf{u}} A_\mathbf{0} T_{\mathbf{u}}^{\dagger}.    
\end{equation}
The discrete Wigner function of a state $\rho$ at the point  $\mathbf{u}$ is then defined as 
\begin{equation}
    W_\rho(\mathbf{u}):=\frac{1}{d} \tr\left[A_{\mathbf{u}} \rho\right].
\end{equation}
More generally, we can replace $\rho$ with $H$ for the discrete Wigner function of a Hermitian operator $H$. For the case of $H$ being an effect $E$ of some Positive Operator-Valued Measure (POVM), its discrete Wigner function is given by 
\begin{equation}
    W(E|\mathbf{u}) := \tr[EA_\mathbf{u}].
\end{equation} 
There are several useful properties of the set $\{A_{\bu}\}_{\bu}$ as follows:
\begin{enumerate}
    \item $A_{\bu}$ is Hermitian;
    \item $\sum_{\bu} A_{\bu}/d = \mathds{1}$;
    \item $\text{Tr}[A_{\bu}A_{{\bu}'}] = d \delta({\bu}, {\bu}')$;
    \item $\text{Tr}[A_{\bu}] = 1$;
    \item $H = \sum_{\bu} W_{H}(\bu)A_{\bu}$;
    \item $\{ A_{\bu}\}_{\bu} = \{ {A_{\bu}}^T\}_{\bu}$.  
\end{enumerate}
We say a Hermitian operator $H$ has positive discrete Wigner functions (PWFs) if $W_H(\mathbf{u}) \geq 0, \forall \mathbf{u} \in \mathbb{Z}_d \times \mathbb{Z}_d$. According to the discrete Hudson’s theorem~\cite{gross2006hudson}, a pure state $\rho$ is a stabilizer state, if and only if it has PWFs. Similarly, an $n$-valued POVM $\mathbf{E} = \{ E_j\}_{j=0}^{n-1}$ is said to be a PWF POVM if each $E_j$ has PWFs.

The discrete Wigner function of each measurement outcome of a POVM $\{E_j \}_{j=0}^{n-1}$ has a conditional quasi-probability interpretation over the phase space
\begin{equation}
     \sum_{j}W(E_j|\mathbf{u}) = 1,
\end{equation}
where $E_j \geq \mathbf{0}$ and $\sum_j E_j = \mathds{1}$. In the case of ${E_j}$ having PWFs, $W(E_j|\mathbf{u})$ can be interpreted as the probability of obtaining outcome $j$ given that the system is at the phase space point $\mathbf{u}$. This property is crucial for efficiently simulating quantum computation classically.
The total probability of obtaining outcome $j$ from a measurement on state $\rho$ is then given by
\begin{equation}
P(j|\rho) = \sum_\mathbf{u} W_{\rho}(\mathbf{u})W(E_j|\mathbf{u}),
\end{equation}
where $P(j|\rho)$ can be effectively estimated~\cite{Pashayan_2015, Mari2012} when both $\rho$ and $E_j$ have PWFs, implying that both $W_{\rho}(\mathbf{u})$ and $W(E_j|\mathbf{u})$ possess classical probability interpretations. Therefore, negative quasi-probability is a vital resource for quantum speedup in stabilizer computation and has deep connections with contextuality in stabilizer measurements~\cite{Veitch2012,Howard_2014}. In this sense, PWF POVMs are regarded as classically-simulable measurements~\cite{Howard_2014}, which strictly include all stabilizer measurements.

\begin{lemma}\label{lem:PSP_Op_Eigenvale_d}
    For any phase-space point operator $A_\mathbf{u} \in \cL(\cH_d)$, $ \mathbf{u} \in \mathbb{Z}_d \times \mathbb{Z}_d$, $A_\mathbf{u}$ is a unitary operator with eigenvalues of $+1$ or $-1$, where eigenvalue +1 has a degeneracy of $\frac{d+1}{2}$ and eigenvalue $-1$ has a degeneracy of $ \frac{d-1}{2}$.
\end{lemma}
\begin{proof}
Suppose $A_{\mathbf{0}}$ has a spectral decomposition
\begin{equation}
    A_{\mathbf{0}} = \sum_i a_i\ket{a_i}\bra{a_i}.
\end{equation}
Since $A_\mathbf{u} = T_\mathbf{u}A_{\mathbf{0}}T_\mathbf{u}^{\dagger}$, we have $A_\mathbf{u} = \sum_i a_i T_\mathbf{u}\ket{a_i}\bra{a_i}T_\mathbf{u}^{\dagger}$ which means all possible $A_\mathbf{u}$ have same eigenvalues $\{a_i\}$ with corresponding eigenvectors $\{T_\mathbf{u}\ket{a_i}\}$. Note that $A_{\mathbf{0}} =  \sum_{k \in \mathbb{Z}_d} \ketbra{k}{-k}$~\cite{emeriau2022quantum}, we conclude that the matrix representation of $A_{\mathbf{0}}$ is $A_{\mathbf{0}} = \begin{bmatrix} 1 &0\\0& {\sigma_x}^{\otimes \lfloor \frac{d}{2} \rfloor}
\end{bmatrix}$, where $\sigma_x = \begin{bmatrix} 0 &1\\1&0
\end{bmatrix}$.
Now we consider the eigenvalues of $A_{\mathbf{0}}$. Notice that 

\begin{equation}
\det(|aI_d-A_{\mathbf{0}}|) = \det\left(\left|aI_d-\begin{bmatrix} 1 &0\\0& {\sigma_x}^{\otimes \lfloor \frac{d}{2} \rfloor}
\end{bmatrix}\right|\right) = 
(1-a)\det\left|aI_{d-1}- {\sigma_x}^{\otimes \lfloor \frac{d}{2} \rfloor}\right| = (1-a)(1-a^2)^{\lfloor \frac{d}{2} \rfloor} = 0,
\end{equation}

where $a$ denotes the eigenvalue of $A_{\mathbf{0}}$. Thus, eigenvalues of $A_\mathbf{u}$ are $+1$ or $-1$, and eigenvalue $+1$ has a degeneracy of $\frac{d+1}{2}$ and eigenvalue $-1$ has a degeneracy of $ \frac{d-1}{2}$ due to $\tr(A_\mathbf{u}) = 1$. We can further conclude that $A_\mathbf{u}$ is unitary according to the possible eigenvalues of $A_\mathbf{u}$.
\end{proof}

\section{Asymptotic limits of PWF POVMs}\label{appendix:asym_limits_nmPOVM}
\renewcommand\theproposition{\ref{lem:unext_to_strongly_unext}}
\setcounter{proposition}{\arabic{proposition}-1}
\begin{lemma} 
For a PWF unextendible subspace $\cS \subseteq \cH_d$, if there is a PWF state $\rho \in \cD(\cS)$ such that $\mathrm{supp}(\rho) = \cS$, then $\cS$ is strongly PWF unextendible.
\end{lemma}

\begin{proof}
First, we will demonstrate that $\cS^{\ox 2}$ is PWF unextendible through a proof by contradiction. Suppose $\rho_s\in\cS$ is a PWF state such that $\mathrm{supp}(\rho_s) = \mathrm{supp}(\cS)$. If there is a PWF state $\sigma$ supporting on $(\cS^{\ox 2})^\perp$, then we have $\tr[\sigma(\rho_s\ox \rho_s)] = 0$ which leads to 
\begin{equation} 
    \tr\Big[\rho_s\tr_2[\sigma(\mathds{1}\ox \rho_s )] \Big] = 0.
\end{equation}
Now we construct an operator $\sigma'\in \cL(\cH_d)$ by \begin{equation}
    \sigma'=\tr_2[\sigma(\mathds{1}\ox \rho_s )].
\end{equation}
It is easy to check that $\sigma'$ is hermitian, $\sigma' \geq \mathbf{0}$ and $\tr(\sigma'\rho_s) = 0$.

If $\sigma' = \mathbf{0}$, we know that $\tr\sigma' = 0$ which indicates that $\tr[\sigma(\mathds{1}\ox \rho_s)] = \tr[\rho_s\tr_1\sigma] = 0$. We note that $\tr_1\sigma \neq \mathbf{0}$ otherwise $\sigma= \mathbf{0}$. Also, $\tr_1\sigma$ is PWF because partial trace preserves the positivity of the discrete Wigner functions which can be observed by expressing the state as $\sigma = \sum_{\bu} W_{\sigma}(\mathbf{u})A_{\bu}$. Thus, we will get a PWF state supporting on $\cS^\perp$ after normalizing $\tr_1\sigma$, a contradiction to the PWF unextendibility of $\cS$.

If $\sigma' \neq \mathbf{0}$, we can calculate the Wigner functions of $\sigma'$ and demonstrate their non-negativity as follows.
\begin{align}
\label{eq:new_POVM1}
W_{\sigma'}(\mathbf{u}_1) &= \frac{1}{d}\tr(\sigma' A_{\mathbf{u}_1})=\frac{1}{d^2}\sum_{\mathbf{u}_2}\tr[\sigma (A_{\mathbf{u}_1}\ox A_{\mathbf{u}_2})] \tr(\rho_s A_{\mathbf{u}_{2}}).
\end{align}
Since $\sigma$ and $\rho_s$ are PWF, i.e., $\tr[\sigma (A_{\mathbf{u}_1}\ox A_{\mathbf{u}_2})]\geq 0, \tr(\rho_s A_{\mathbf{u}_{2}}) \geq 0, \forall \mathbf{u}_1, \mathbf{u}_2$, we have $W_{\sigma'}(\mathbf{u}_1) \geq 0$. Thus $\sigma'$ is PWF. Consequently, we have obtained a PWF state supporting on $\cS^\perp$ after normalizing $\sigma'$, a contradiction to the PWF unextendibility of $\cS$.

Hence, we conclude that $\cS^{\ox 2}$ is PWF unextendible.
Using a similar technique, we can prove that $\cS^{\ox 3}$ is PWF unextendible by making a contradiction to the PWF unextendibility of $\cS^{\ox 2}$. In turn, we can conclusively demonstrate that $\cS^{\ox k}$ is PWF unextendible for any positive integer $k$, which completes the proof. 
\end{proof}

\renewcommand\theproposition{\ref{prop:case_allmagic}}
\setcounter{proposition}{\arabic{proposition}-1}
\begin{proposition}
    There exists a strongly PWF unextendible subspace $\cS\subseteq \cH_d$ of dimension $(d+1)/2$. 
\end{proposition}
\begin{proof}
First, we construct a $(d-1)/2$ dimensional subspace  $\cS_m \subseteq \cH_d$ that supports only magic states. Then we will show that $\cS_m^{\perp} \subseteq \cH_d$ is a strongly PWF unextendible subspace of dimension $(d+1)/2$.
We consider the eigenspace of the phase-space point operator $A_{\mathbf{0}}$. Denote the set of all eigenvectors of $A_{\mathbf{0}}$ corresponding to eigenvalue of $-1$ as $S^-\coloneqq  \{\ket{a_i^-}\}_{i=1}^{\frac{d-1}{2}}$. We will show these states in $S^-$ span a subspace $\cS_m \in \mathcal{H}_d$ that contains no PWF states.

Obviously, any $\ket{a_i^-} \in S^-$ is a magic state due to $\tr(\ket{a_i^-}\bra{a_i^-}A_{\mathbf{0}}) = -1$. Suppose $\ket{\psi}$ is an arbitrary pure state in $S^-$. It can be written as $\ket{\psi} = \sum_i \alpha_i \ket{a_i^-}.$ The Wigner function of $\ket{\psi}$ at the phase-space point $\mathbf{0}$ is 
\begin{align}
  W_{\psi}(\mathbf{0}) &= \frac{1}{d}\bra{\psi}A_{\mathbf{0}}\ket{\psi} = \frac{1}{d} \sum_{i,j} \alpha_i^* \alpha_j \bra{a_i^-}A_{\mathbf{0}}\ket{a_j^-} = -\frac{1}{d} \sum_{i,j} \alpha_i^* \alpha_j \braket{a_i^-}{a_j^-} = -\frac{1}{d} \sum_i \alpha_i \alpha_i^*= -\frac{1}{d},
\end{align}
which tells $\ket{\psi}$ is a magic state. For any mixed state $\rho = \sum_i p_i \ketbra{\psi_i}{\psi_i}$ on $ \mathcal{S}_m$, we have
\begin{equation}
W_{\rho}(\mathbf{0}) = \sum_i p_i W_{\psi_i}(\mathbf{0})= -\frac{1}{d} \sum_i p_i= -\frac{1}{d}.
\end{equation}
Thus, we construct a $(d-1)/2$ dimensional subspace  $\cS_m$ that contains no PWF states. Obviously, $\cS_m^{\perp}$ is a PWF unextendible subspace of dimension $(d+1)/2$. $\cS_m^{\perp}$ is spanned by the set of all eigenvectors of $A_{\mathbf{0}}$ corresponding to eigenvalue of $+1$, denoted as $S^+\coloneqq  \{\ket{a_i^+}\}_{i=1}^{\frac{d+1}{2}}$. We show that $\rho_n = \frac{2}{d+1}\sum_j \ketbra{a_j^+}{a_j^+}$ is a PWF state on $\cS_m^{\perp}$ as follows: 
\begin{subequations}
    \begin{align}
        W_{\rho_n}(\mathbf{u}) &= \frac{1}{d} \tr(A_\mathbf{u}\rho_n) \\
        & = \frac{2}{(d+1)d} \tr(A_\mathbf{u}\sum_j \ketbra{a_j^+}{a_j^+} )  \label{eq:s16b}\\
        & = \frac{2}{(d+1)d} \tr[A_\mathbf{u}(I+A_\mathbf{0})/2] \label{eq:s16c}\\
        & =  \frac{1+ \delta_{\mathbf{u},\mathbf{0}}}{(d+1)} > 0.
    \end{align}
\end{subequations}
From Eq.~\eqref{eq:s16b} to Eq.~\eqref{eq:s16c}, we use the properties that $A_{\mathbf{0}} = \sum_j \ketbra{a_j^+}{a_j^+} - \sum_{i} \ketbra{a_i^-}{a_i^-}$ with spectral decomposition, and $I_d = \sum_j \ketbra{a_j^+}{a_j^+} + \sum_{i} \ketbra{a_i^-}{a_i^-}$.
Note that $\mathrm{supp}(\rho_n) = \cS_m^{\perp}$, combined with Lemma~\ref{lem:unext_to_strongly_unext}, we can conclude that $\cS_m^{\perp} \subseteq \cH_d$ is a strongly PWF unextendible subspace of dimension $(d+1)/2$.
\end{proof}

\renewcommand\theproposition{\ref{thm:criterion_nm}}
\setcounter{proposition}{\arabic{proposition}-1}
\begin{theorem}
    Given $\rho_0,\rho_1\in \cD(\cH_d)$, if any of them has strictly positive discrete Wigner functions, i.e., $W_{\rho_i}(\mathbf{u}) >0, \forall \mathbf{u}$, then for any integer $n\in \mathbb{Z}^+$, $\rho_0^{\ox n}$ and $\rho_1^{\ox n}$ cannot be unambiguously distinguished by PWF POVMs.
\end{theorem}

\begin{proof}
    Suppose the state $\rho_0$ and $\rho_1$ can be unambiguously distinguished by a PWF POVM $\{E_0, E_1\}$. By definition, we have
\begin{equation}\label{Eq:cond_perf}
    \tr(E_0\rho_1) = 0 \text{ and } \tr(E_1\rho_0) = 0.
\end{equation}
Then we are going to establish the theorem using a proof by contradiction. Without loss of generality, we suppose $\rho_1$ has strictly positive Wigner functions. Notice that
\begin{equation}\label{Eq:trErho}
    \tr(E_0\rho_1) = \sum_{\mathbf{u}}W(E_0|\mathbf{u}) W_{\rho_1}(\mathbf{u}) = 0.
\end{equation}
By the strictly positivity of the Wigner functions of $\rho_1$, i.e., $W_{\rho_1}(\mathbf{u})>0, \forall \mathbf{u}$, we have that Eq.~\eqref{Eq:trErho} holds if and only if $W(E_0|\mathbf{u}) =0, \forall \mathbf{u}$. Combining the fact that $\sum_{\mathbf{u}}W(E_0|\mathbf{u}) = d\tr(E_0)$, we have $\tr(E_0) = 0$. It follows that all eigenvalues of $E_0$ are equal to zero since $E_0\geq 0$. Then we have $E_0 = \mathbf{0}$ which gives $\tr(E_1\rho_0) = \tr(\mathds{1}\rho_0) = 1$, a contradiction. 

Hence, there is no effect $E_0$ having PWFs such that $\tr(E_0\rho_0) > 0$ and $\tr(E_0\rho_1)=0$ if $\rho_1$ has strictly positive Wigner functions. Similarly, we can show there is no effect $E_1$ having PWFs such that $\tr(E_1\rho_0) = 0$ and $\tr(E_1\rho_1)>0$ if $\rho_0$ has strictly positive Wigner functions. Using the fact that 
\begin{equation}
W_{\rho^{\ox 2}}(\mathbf{u}_i \oplus \mathbf{u}_j) = W_{\rho}(\mathbf{u}_i)W_{\rho}(\mathbf{u}_j), \quad \forall \mathbf{u}_i, \mathbf{u}_j \in \mathbb{Z}_d \times \mathbb{Z}_d,
\end{equation}
we complete the proof.
\end{proof}
\renewcommand{\theproposition}{S\arabic{proposition}}

\section{Minimum error discrimination by PWF POVMs}\label{appendix:dual_sdp_minerr}

Note that given a two-valued PWF POVM $\{E, \mathds{1}-E \}$, the discrete Wigner function of an effect $E$ is $W(E|\mathbf{u}) = \tr(EA_\mathbf{u})$. The SDP of discriminating an equiprobable pair of states $\{\rho_0, \rho_1\}$ via PWF POVMs can be written as 
\begin{equation}\label{Eq:err_sdp}
\begin{aligned}
P_{\mathrm{e}}^{{\textup{\tiny PWF}}}(\rho_0, \rho_1, \frac{1}{2}) = \min_{E} &\; \frac{1}{2} + \frac{1}{2}\tr[E(\rho_1 - \rho_0)],\\
     {\rm s.t.} & \;\; 0\leq E \leq \mathds{1}, \\ 
                & \;\; 0\leq \tr[E A_\mathbf{u}] \leq 1, \forall \ \mathbf{u},
\end{aligned}
\end{equation}
where $E \leq \mathds{1}$ implies $\mathds{1} - E$ is positive semidefinite. For different linear inequality constraints, we introduce corresponding dual variables $V$, $U$, $a_\mathbf{u}, b_{\mathbf{u}} \geq 0$. Then the Lagrange function of the primal problem can be written as
\begin{equation}
\begin{aligned}
L(E,V,U,a_\mathbf{u},b_\mathbf{u}) &= \frac{1}{2} + \frac{1}{2}\tr[E(\rho_1-\rho_0)] + \tr[V(E-\mathds{1})]- \tr(UE) \\
&\quad- \sum_{\mathbf{u}} a_{\mathbf{u}} \tr(E A_{\mathbf{u}}) + \sum_{\mathbf{u}} b_{\mathbf{u}} [\tr(EA_{\mathbf{u}})-1]\\
&= \frac{1}{2} + \tr\Big[E\Big(V- U+ \frac{1}{2}(\rho_1-\rho_0) - \sum_{\mathbf{u}} a_{\mathbf{u}}A_{\mathbf{u}} + \sum_{\mathbf{u}} b_{\mathbf{u}}A_{\mathbf{u}}\Big) \Big] - \tr(V)- \sum_{\mathbf{u}} b_{\mathbf{u}}
\end{aligned}
\end{equation}
The corresponding Lagrange dual function is
\begin{equation}
    g(V,U,a_\mathbf{u},b_\mathbf{u}) = \inf_{E} L(E,V,U,a_\mathbf{u},b_\mathbf{u}).
\end{equation}
We can see that $V- U+ \frac{1}{2}(\rho_1-\rho_0) - \sum_{\mathbf{u}} a_{\mathbf{u}}A_{\mathbf{u}} + \sum_{\mathbf{u}} b_{\mathbf{u}}A_{\mathbf{u}} \geq 0$, otherwise $g(V,U,a_\mathbf{u},b_\mathbf{u})$ is unbounded. Thus the dual SDP is
\begin{equation}\label{Eq:dual_sdp_err}
\begin{aligned}
\max_{V,U,a_{\mathbf{u}},b_{\mathbf{u}}} &\; \frac{1}{2}-\tr(V)- \sum_{\mathbf{u}} b_{\mathbf{u}},\\
    {\rm s.t.} & \;\; U\geq 0, V\geq 0,\\
    & \;\; V - U + \frac{1}{2}(\rho_1-\rho_0) \geq  \sum_{\mathbf{u}} (a_{\mathbf{u}}-b_{\mathbf{u}}) A_{\mathbf{u}}, \\ 
    & \;\; a_{\mathbf{u}}\geq 0, b_{\mathbf{u}}\geq 0, \quad \forall \ \mathbf{u}.
\end{aligned}
\end{equation}

\renewcommand\theproposition{\ref{prop:s_state_err}}
\setcounter{proposition}{\arabic{proposition}-1}
\begin{proposition}
Let $\rho_0$ be the Strange state $\ketbra{\mathbb{S}}{\mathbb{S}}$ and $\rho_1 = (\mathds{1}-\ketbra{\mathbb{S}}{\mathbb{S}})/2$ be its orthogonal complement. For $n\in \mathbb{Z}^+$, we have 
\begin{equation}
    P_{\mathrm{e}}^{{\textup{\tiny PWF}}}(\rho_0^{\ox n}, \rho_1^{\ox n}, \frac{1}{2}) = \frac{1}{2^{n+1}}.
\end{equation}
The optimal PWF POVM is $\{E, \mathds{1}-E \}$, where
$E = (\ketbra{\mathbb{K}}{\mathbb{K}} + \ketbra{\mathbb{S}}{\mathbb{S}})^{\ox n}$ and $\ket{\mathbb{K}} = (\ket{1}+\ket{2})/\sqrt{2}$.
\end{proposition}
\begin{proof}
First, we are going to prove $P_{\mathrm{e}}^{{\textup{\tiny PWF}}}(\rho_0^{\ox n}, \rho_1^{\ox n}, \frac{1}{2}) \leq \frac{1}{2^{n+1}}$ using SDP~\eqref{Eq:err_sdp}. We will show that $E = (\ketbra{\mathbb{K}}{\mathbb{K}} + \ketbra{\mathbb{S}}{\mathbb{S}})^{\ox n}$ is a feasible solution with a discrimination error $\frac{1}{2^{n+1}}$. In specific, it is easy to check $0\leq E\leq \mathds{1}$. Furthermore, we can check that $\ket{0}, \ket{\mathbb{K}}$ and $\ket{\mathbb{S}}$ are eigenvectors of $A_\mathbf{0}$ with eigenvalue $+1, +1$ and $-1$, respectively. It follows
\begin{equation}
\tr[A_{\mathbf{u}}(\ketbra{\mathbb{K}}{\mathbb{K}} + \ketbra{\mathbb{S}}{\mathbb{S}})] = \tr[A_{\mathbf{u}}(\mathds{1}-\ketbra{0}{0})] = 1-\tr(A_{\mathbf{u}}\ketbra{0}{0})\geq 0,    
\end{equation}
where the inequality is due to the fact that $A_{\mathbf{u}}$ has eigenvalues no larger than $1$. Also, we have $\tr[A_{\mathbf{u}}(\ketbra{\mathbb{K}}{\mathbb{K}} + \ketbra{\mathbb{S}}{\mathbb{S}})]=1-\tr(A_{\mathbf{u}}\ketbra{0}{0})\leq 1$ as $\ketbra{0}{0}$ is a stabilizer state with $\tr(A_{\mathbf{u}}\ketbra{0}{0})\geq 0$. Thus, for the $n$-copy case, we have
\begin{equation}
    0\leq \prod_{i=1}^n \Big(\bra{\mathbb{K}}A_{\mathbf{u}_i}\ket{\mathbb{K}} + \bra{\mathbb{S}}A_{\mathbf{u}_i}\ket{\mathbb{S}} \Big) \leq 1,
\end{equation}
which makes $E$ satisfies $0\leq \tr[EA_\mathbf{u}]\leq 1$. Hence, $E$ is a feasible solution to the primal SDP~\eqref{Eq:err_sdp}. Note that 
\begin{subequations}
\begin{align}
    \tr[(\ketbra{\mathbb{K}}{\mathbb{K}} + \ketbra{\mathbb{S}}{\mathbb{S}})\rho_0] &= \braket{\mathbb{K}}{\mathbb{S}}\braket{\mathbb{S}}{\mathbb{K}} + \braket{\mathbb{S}}{\mathbb{S}}\braket{\mathbb{S}}{\mathbb{S}}=1,\\ 
    \tr[(\ketbra{\mathbb{K}}{\mathbb{K}} + \ketbra{\mathbb{S}}{\mathbb{S}})\rho_1] &= \frac{1}{2}\bra{\mathbb{K}}(\mathds{1}-\ketbra{\mathbb{S}}{\mathbb{S}})\ket{\mathbb{K}} + \frac{1}{2}\bra{\mathbb{S}}(\mathds{1}-\ketbra{\mathbb{S}}{\mathbb{S}})\ket{\mathbb{S}}=\frac{1}{2}.
\end{align}
\end{subequations}
The corresponding discrimination error is
\begin{subequations}
\begin{align}
    P_{pr}^* &= \frac{1}{2} + \frac{1}{2}\tr\Big[(\ketbra{\mathbb{K}}{\mathbb{K}} + \ketbra{\mathbb{S}}{\mathbb{S}})^{\ox n}(\rho_1^{\ox n} - \rho_0^{\ox n})\Big]\\
    &= \frac{1}{2} + \frac{1}{2}\tr\Big[(\ketbra{\mathbb{K}}{\mathbb{K}}\rho_1 + \ketbra{\mathbb{S}}{\mathbb{S}}\rho_1)^{\ox n} - (\ketbra{\mathbb{K}}{\mathbb{K}}\rho_0 + \ketbra{\mathbb{S}}{\mathbb{S}}\rho_0)^{\ox n}\Big]\\
    &= \frac{1}{2} + \frac{1}{2}\left[\frac{1}{2^{n}} - 1\right] = \frac{1}{2^{n+1}}.
\end{align}
\end{subequations}

Second, we use the dual SDP~\eqref{Eq:dual_sdp_err} to show $P_{\mathrm{e}}^{{\textup{\tiny PWF}}}(\rho_0^{\ox n}, \rho_1^{\ox n}, \frac{1}{2}) \geq \frac{1}{2^{n+1}}$. We will construct a valid $a_\mathbf{u}$ combined with $\{V = (2^n-1)\rho_0/2^{n+1}, U = 0, b_{\mathbf{u}} = 0\}$ as a feasible solution to the dual problem. We note that $\rho_0 = (\mathds{1}- A_\mathbf{0})/2$ and $\rho_1 = (\mathds{1} + A_\mathbf{0})/4$ and introduce the following notation. Let $\mathbf{k} = (k_1, k_2, ..., k_n)\in \{0,1\}^n$ be a $n$-bit binary string and $|\mathbf{k}|$ be the Hamming weight of it. We then denote $A_{\mathbf{k}} = A_{k_1}\ox A_{k_2}\ox ... \ox A_{k_n}$ where $A_{k_i} = A_\mathbf{0}$ if $k_i = 1$ and $A_{k_i} = \mathds{1}$ if $k_i = 0$. Then we have
\begin{subequations}
\begin{align}
    V-U + \frac{1}{2}(\rho_1^{\ox n} - \rho_0^{\ox n}) &= \frac{2^n-1}{2^{n+1}}\left(\frac{\mathds{1} - A_\mathbf{0}}{2}\right)^{\ox n} + \frac{1}{2}\left[\left(\frac{\mathds{1} + A_\mathbf{0}}{4}\right)^{\ox n} - \left(\frac{\mathds{1}-A_\mathbf{0}}{2}\right)^{\ox n}\right]\\
    &= \frac{1}{2^{2n+1}}(\mathds{1} + A_\mathbf{0})^{\ox n} - \frac{1}{2^{2n+1}}(\mathds{1} - A_\mathbf{0})^{\ox n}\\
    &= \frac{1}{2^{2n+1}}\sum_{\mathbf{k}\in \{0,1\}^n}\Big(1 - (-1)^{|\mathbf{k}|}\Big) A_{\mathbf{k}}\label{Eq:IeqSUM1}\\
    &= \frac{1}{2^{2n+1}}\sum_{\mathbf{k}\in \{0,1\}^n}\Big(\frac{1 - (-1)^{|\mathbf{k}|}}{3^{n-|\mathbf{k}|}} \sum_{\mathbf{u}_{\mathbf{k}}} A_{\mathbf{u}_{\mathbf{k}}}\Big)\label{Eq:IeqSUM2},
\end{align}
\end{subequations}
where $A_{\mathbf{u}_{\mathbf{k}}} = A_{\mathbf{u}_1}\ox A_{\mathbf{u}_2}\ox \cdots \ox A_{\mathbf{u}_n}$ with $\mathbf{u}_j = \mathbf{0}$ if $k_j=1$ for $j=1,2...,n$. To derive Eq.~\eqref{Eq:IeqSUM2} from Eq.~\eqref{Eq:IeqSUM1}, we express each $A_{k_i} = \mathds{1}$ with $k_i = 0$ in $A_{\mathbf{k}}$ as $\mathds{1} = \frac{1}{3}\sum_{\mathbf{u}} A_{\mathbf{u}}$, where each $A_{\mathbf{k}}$ contains $(n-|\mathbf{k}|)$ occurrences of $\mathds{1}$. Thus, we can find a set of $\hat{a}_\mathbf{u}$ such that 
\begin{equation}
    V-U + \frac{1}{2}(\rho_1^{\ox n} - \rho_0^{\ox n}) = \sum_{\mathbf{u}} \hat{a}_{\mathbf{u}} A_{\mathbf{u}}, 
\end{equation}
by the following argument. For each $A_\mathbf{u'}$ in the $n$-copy system, we may find it as the sum of some terms in Eq.~\eqref{Eq:IeqSUM2} with all coefficient positive since $\frac{1 - (-1)^{|\mathbf{k}|}}{3^{n-{|\mathbf{k}|}}}\geq 0$. We can then let $\hat{a}_\mathbf{u}$ be the sum of those coefficients, which makes $\{V = (2^n-1)\rho_0/2^{n+1}, \hat{a}_\mathbf{u}, U = 0, b_{\mathbf{u}} = 0\}$ a feasible solution of the dual SDP. Thus we have  
\begin{equation}
    P_{du}^* = \frac{1}{2} - \tr(V) = \frac{1}{2^{n+1}}.
\end{equation}
Combining it with the primal part and utilizing Slater's condition for strong duality~\cite{boyd2004convex}, we have that $P_{\mathrm{e}}^{{\textup{\tiny PWF}}}(\rho_0^{\ox n}, \rho_1^{\ox n}, \frac{1}{2}) = \frac{1}{2^{n+1}}$.
\end{proof}
\renewcommand{\theproposition}{S\arabic{proposition}}

Notice that for a given measurement $\mathbf{M}:=\{M_j\}_j$, we define the PWF robustness of measurement as
\begin{equation}
    \mathbf{R}_{\tiny {\mathcal{E}_{\text{PWF}}}}(\mathbf{M}) = \left.\min\Big\{r\in\mathbb{R}_+ \right| M_j+rN_j \in \mathcal{E}_{\text{PWF}}\, \forall j, \{N_j\}_j \in \mathcal M \Big\},
\end{equation} 
where we denote by $\mathcal M$ the set of all possible POVMs, and denote by $\mathcal{E}_{\text{PWF}}$ the set of all PWF effects. An effect $E$ belongs to $\mathcal{E}_{\text{PWF}}$ if it has PWFs. The \textit{data-hiding ratio}~\cite{Lami_2018} associated with PWF POVMs is defined in our manuscript as
\begin{equation}\label{Eq:dhratio} 
    R(\text{PWF}) = \max \frac{\|p\rho-(1-p)\sigma\|_{\tiny \text{All}}}{\|p\rho-(1-p)\sigma\|_{\tiny \text{PWF}}},
\end{equation}
where the maximization ranges over all pairs of states $\rho,\sigma$ and a priori probabilities $p$ (here we also define $\|\cdot\|_{\tiny \text{PWF}}$ as the \textit{distinguishability norm} associated with PWF POVMs). In an intuitive sense, we could imagine that a higher data-hiding ratio in Eq.~\eqref{Eq:dhratio} will be obtained if the optimal POVMs for $\|\cdot\|_{\textbf{All}}$ exhibit `less PWF'. This would suggest a more pronounced disparity allowing the agent to access the optimal discrimination strategy without a `magic factory' in the given physical setting. Therefore, given an equiprobable pair of states $\{\rho,\sigma\}$, we define $\mathbf{R}^*_{\mathcal{E}_{\text{PWF}}}(\mathbf{M}_{\rho,\sigma})$ as the minimum PWF robustness of measurement that an optimal POVM must have to discriminate $\{\rho,\sigma\}$. It can be computed via the following SDP 
\begin{subequations}
\begin{align}
\mathbf{R}^*_{\mathcal{E}_{\text{PWF}}}(\mathbf{M}_{\rho,\sigma}) = \min &\;\; r\\
 {\rm s.t.} & \; E_0, E_1, N_0, N_1 \geq 0, \\
            & \; E_0+E_1=\mathds{1}, N_0 + N_1 = r\cdot \mathds{1},\\
            & \; \tr\left[ (\rho - \sigma) E_0\right] = \frac{1}{2}\| \rho - \sigma\|_1,\label{Eq:optPOVM}\\
            & \; W(E_0+N_0 | \mathbf{u}) \geq 0, W(E_1+N_1 | \mathbf{u}) \geq 0, \forall \mathbf{u},\label{Eq:pwf}
\end{align} 
\end{subequations} 
where the constraint in Eq.~\eqref{Eq:optPOVM} ensures that optimal discrimination is achieved, and the constraints in Eq.~\eqref{Eq:pwf} ensure that $E_j+rN_j\in \text{PWF}$. We generate 500 equiprobable pair of states $\{\rho_j,\sigma_j\}_{j=1}^{500}$ where $\rho_j$ is a random pure qutrit state according to the Haar measure and $\sigma_j$ is its orthogonal complement. Then we compute the ratio $R^*(\text{PWF},\{\rho,\sigma\}) = \|\frac{1}{2}\rho_j-\frac{1}{2}\sigma_j\|_{\tiny \text{All}}/\|\frac{1}{2}\rho-\frac{1}{2}\sigma\|_{\tiny \text{PWF}}$ and $\mathbf{R}^*_{\mathcal{E}_{\text{PWF}}}(\mathbf{M}_{\rho_j,\sigma_j})$. The numerical calculations are implemented in MATLAB~\cite{MATLAB} with the interpreters CVX~\cite{cvx} and QETLAB~\cite{qetlab}. The results are depicted as follows.
\begin{figure}[H]
    \centering
    \includegraphics[width=0.5\linewidth]{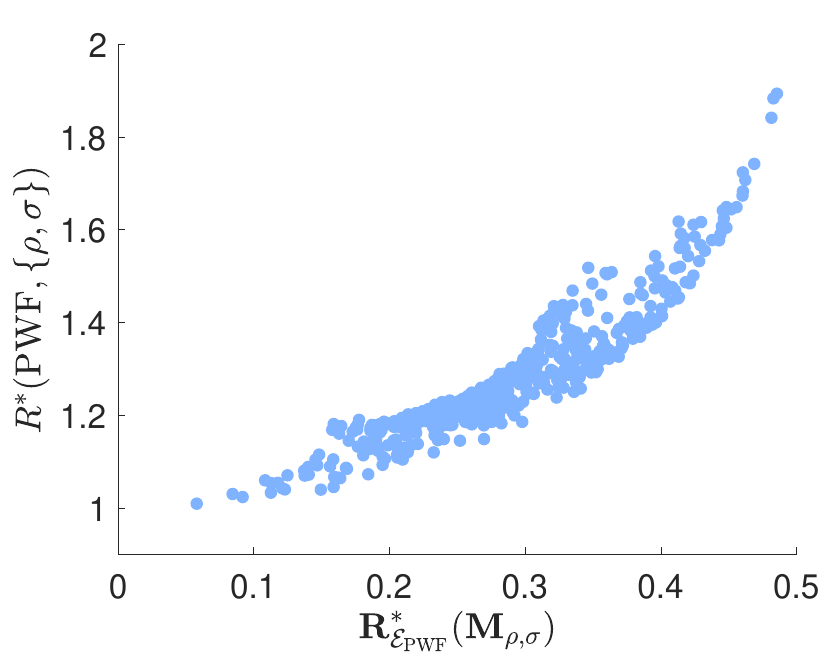}
    \label{fig:Rob_DHR} 
\end{figure}
We observe that there is a possible correlation between the PWF robustness of measurement and the data-hiding ratio associated with a state pair: as the optimal POVM for a state pair exhibits a higher PWF robustness, the corresponding data-hiding ratio also increases. However, their specific relationship remains unclear so far. This experiment also indicates that the data-hiding ratio associated with the Strange state and its orthogonal complement is already relatively high, which equals $2$ when $n=1$ as stated in Eq.~\eqref{eq:data_hiding_PWF} in our manuscript. A deeper relationship between the data-hiding ratio and the PWF robustness of measurement in the case of PWF POVMs merits further investigation.

\section{Distinctions between the QRT of magic states and entanglement}\label{appendix:distinctions}

\renewcommand\theproposition{\ref{thm:absence_UPB}}
\setcounter{proposition}{\arabic{proposition}-1}
\begin{theorem}
For a subspace $\cS\in \cH_d$, if $\cS$ has a set of basis $\{\ket{\psi_i}\}_{i=1}^n$ where every $\ket{\psi_i}$ is a stabilizer state, then $\cS$ is PWF extendible.
\end{theorem}
\begin{proof}
Since $\{\ket{\psi_i}\}_{i=1}^n$ is a basis for $\cS$, we have $\ket{\psi_i}$ and $\ket{\psi_j}$ are orthogonal which yields
\begin{equation}\label{Eq:innerprod}
    \langle\psi_i |\psi_j\rangle = \sum_{\bu}W_{\psi_i}(\bu)W_{\psi_j}(\bu) = 0.
\end{equation}
for any $i\neq j$. Note that every pure stabilizer state has Wigner functions $0$ or $1/d$~\cite{gross2006hudson}. Then we know that for a fixed point $\bu'$, there is at most one state $\ket{\psi_{j'}}$ that has $W_{\psi_{j'}}(\bu') = 1/d$. For any other states $\ket{\psi_i}, i\neq j'$, we have $W_{\psi_i}(\bu') = 0$ otherwise $\braket{\psi_{j'}}{\psi_i} \geq 1/d^2>0$, a contradiction to Eq.~\eqref{Eq:innerprod}. Thus, we have $\sum_{i=1}^n W_{\psi_i}(\mathbf{u}) = 0$ or $\sum_{i=1}^n W_{\psi_i}(\mathbf{u}) = 1/d$. Then we denote $P_{\cS} = \sum_{i=1}^n \ketbra{\psi_i}{\psi_i}$ as the projection of $\cS$ and consider its orthogonal complement $P_{\cS}^\perp = \mathds{1}-P_{\cS}$. Considering $P_{\cS}^\perp$ as an effect of the POVM $\{P_{\cS}^\perp, P_{\cS} \}$, we have
\begin{equation}
    W(P_{\cS}^\perp|\mathbf{u}) = 1 - d\sum_{i=1}^n W_{\psi_i}(\mathbf{u}) = 1 \text{ or } 0,
\end{equation}
which shows that $P_{\cS}^\perp$ has PWFs. After normalization, we can obtain a PWF state supported on $\cS^{\perp}$, which indicates that $\cS$ is PWF extendible.
\end{proof}

\renewcommand\theproposition{\ref{prop:magic_cost_sstate}}
\setcounter{proposition}{\arabic{proposition}-1}
\begin{proposition}
Let $\rho_0$ be the Strange state $\ketbra{\mathbb{S}}{\mathbb{S}}$ and $\rho_1 = (\mathds{1}-\ketbra{\mathbb{S}}{\mathbb{S}})/2$ be its orthogonal complement. $\rho_0 \otimes \tau^{\otimes k}$  and $\rho_1 \otimes \tau^{\otimes k}$ cannot be perfectly distinguished for any qutrit magic state $\tau$ and $k =1$ or $2$.

\end{proposition}
\begin{proof}
First, suppose there is a PWF POVM $\{E, \mathds{1}-E\}$ that can perfectly distinguish $\rho_0 \otimes \tau^{\otimes k}$ and $\rho_1 \otimes \tau^{\otimes k}$. Then we have
\begin{equation}\label{Eq:s_cost}
    \tr[(\rho_0 \otimes \tau^{\otimes k})E] = 1, \tr[(\rho_1 \otimes \tau^{\otimes k})E] = 0.
\end{equation}
We can write $\tr[(\rho_0 \otimes \tau^{\otimes k})E] = \tr[\rho_0 \tr_2 [(\mathds{1}\otimes \tau^{\otimes k})E]] = 1$. Notice the fact that when $\rho_0$ is a pure state, $\tr(\rho_0X) = 1, \tr(\rho_1X) = 0$ if and only if $X = \rho_0$. Then for any $\mathbf{u}_1$, we have
\begin{equation}\label{Eq:s_cost3}
    \tr(\rho_0 A_{\mathbf{u_1}}) =  \frac{1}{d^k}\sum_{\mathbf{u_2},\cdots, \mathbf{u_{k+1}}}\tr(EA_{\mathbf{u_1},\cdots, \mathbf{u_{k+1}}}) \tr(\tau A_{\mathbf{u_2}}) \cdots \tr(\tau A_{\mathbf{u_{k+1}}}).
\end{equation}
Suppose the value of $\text{maxneg}(\rho_0)$ is obtained at phase point $\mathbf{u_1} = (a,b)$ for $\rho_0$, where $\text{maxneg}(\rho) \coloneqq -\min_{\mathbf{u}}W_{\rho}(\mathbf{u})$ denotes the maximal negativity of $\rho$.
We consider the right hand of Eq.~\eqref{Eq:s_cost3} by choosing $\mathbf{u_1} = (a,b)$:
\begin{align}
    -d\cdot \text{maxneg}(\rho_0) = &\frac{1}{d^k}\sum_{\mathbf{u_2},\cdots, \mathbf{u_{k+1}}}\tr(EA_{(a,b),\cdots, \mathbf{u_{k+1}}}) \tr(\tau A_{\mathbf{u_2}}) \cdots \tr(\tau A_{\mathbf{u_{k+1}}})\\
    \geq &
    \max(\tr(EA_{(a,b),\cdots, \mathbf{u_{k+1}}})) \frac{1}{d^k} \sum^{< 0}_{\mathbf{u_2},\cdots, \mathbf{u_{k+1}}} \tr(\tau A_{\mathbf{u_2}}) \cdots \tr(\tau A_{\mathbf{u_{k+1}}})\\
    &+ \min(\tr(EA_{(a,b),\cdots, \mathbf{u_{k+1}}})) \frac{1}{d^k} \sum^{\geq 0}_{\mathbf{u_2},\cdots, \mathbf{u_{k+1}}} \tr(\tau A_{\mathbf{u_2}}) \cdots \tr(\tau A_{\mathbf{u_{k+1}}}) \\
    \geq & \frac{1}{d^k} \sum^{< 0}_{\mathbf{u_2},\cdots, \mathbf{u_{k+1}}} \tr(\tau A_{\mathbf{u_2}}) \cdots \tr(\tau A_{\mathbf{u_{k+1}}}) \label{Eq:geq_negative_part}\\
    = & - \text{sn}(\tau^{\otimes k}),
\end{align}
where the inequality in~Eq.~\eqref{Eq:geq_negative_part} is due to the fact that $0 \leq W(E|\mathbf{u}) \leq 1$ for the PWF POVM $\{E, \mathds{1}-E\}$ and $\text{sn}(\rho)\coloneqq \sum_{\mathbf{u}:W_{\rho}(\mathbf{u}) < 0}|W_{\rho}(\mathbf{u})|$ denotes the sum negativity of a magic state $\rho$. Thus, we have 
\begin{equation}
    d\cdot \text{maxneg}(\rho_0) \leq \text{sn}(\tau^{\otimes k}).    
\end{equation}
Note that the Strange state $\rho_0 = \ketbra{\mathbb{S}}{\mathbb{S}}$ satisfies $d\cdot \text{maxneg}(\rho_0) = 1$, which implies $\text{sn}(\tau^{\otimes k}) \geq 1$. Since it has been shown that the maximal sum negativity of a qutrit state is $1/3$~\cite{Veitch2014}, we conclude that
\begin{equation}\label{Eq:sn_le1}
    \text{sn}(\tau^{\otimes k})= [(2\text{sn}(\tau)+1)^k-1]/2 \leq [(5/3)^k-1]/2 < 1,
\end{equation}
for any qutrit magic state $\tau$ and $k=1$ or $2$, where we use the composition law of $\text{sn}(\cdot)$ derived by Ref.~\cite{Veitch2014}. Eq.~\eqref{Eq:sn_le1} is in contradiction with the inequality $\text{sn}(\tau^{\otimes k}) \geq 1$. Thus, we complete the proof. 
\end{proof}

Similarly, we can conclude that for the case of Norell state $\rho_0 = \ketbra{\mathbb{N}}{\mathbb{N}}$, where $\ket{\mathbb{N}} = (-\ket{0} + 2\ket{1} - \ket{2})/\sqrt{6}$~\cite{Veitch2014}, $\rho_0 \otimes \tau$ and $\rho_1 \otimes \tau$ cannot be perfectly distinguished by PWF POVMs for any qutrit state $\tau$.

\end{document}